\documentclass[11pt,onecolumn,compsoc,letterpaper,twoside]{IEEEtran}

 \usepackage{verbatim}
 \usepackage{xcolor}
 \usepackage{ulem}
 \usepackage{cancel}

\ifCLASSOPTIONcompsoc
  \usepackage[nocompress]{cite}
\else
  \usepackage{cite}
\fi

\ifCLASSINFOpdf
   \usepackage[pdftex]{graphicx}
\else
  \usepackage[dvips]{graphicx}
\fi

\usepackage{amsmath}
\usepackage{amssymb}

\usepackage{xcolor}
\usepackage{cleveref}

\usepackage{amsthm,lipsum}
\newtheorem{lemma}{Lemma}

\newtheorem{theorem}{Theorem}

\usepackage{algorithmic}

\usepackage{array}

\usepackage{url}

\begin{document}

\title{Fast Methods for Solving  the Cluster Containment Problem for Phylogenetic Networks}

\author{Andreas~D.~M.~Gunawan,
        Bingxin~Lu,
        and~Louxin~Zhang$^*$% <-this % stops a space
\IEEEcompsocitemizethanks{\IEEEcompsocthanksitem A. Gunawan and L. Zhang are with the Department of Mathematics and B. Lu is with the Department of Computer Science, National University of Singapore, Singapore, 119076.
\IEEEcompsocthanksitem $^*$Corresponding author, e-mail: matzlx@nus.edu.sg}
}

% The paper headers
\markboth{}%
{}
\IEEEtitleabstractindextext{%
\begin{abstract}
Genetic and comparative genomic studies indicate that extant 
 genomes are more properly considered to be a fusion product of random mutations over generations   and genomic material transfers between individuals of different
lineages. This has motivated researchers to adopt phylogenetic networks and other general models  to
study genome evolution.
One important problem arising from reconstruction and verification of phylogenetic networks is the cluster containment problem, namely determining whether or not a cluster of taxa is displayed in a phylogenetic network.
 In this work, a new upper bound for this NP-complete problem is established through an efficient reduction to the SAT problem. Two efficient (albeit exponential time) methods are also implemented. 
It is developed on the basis of generalization of the so-called reticulation-visible property of phylogenetic networks.  
\end{abstract}

}

\maketitle

\IEEEdisplaynontitleabstractindextext
\setlength{\IEEEilabelindent}{8pt}
\setlength{\IEEEelabelindent}{8pt}
\setlength{\IEEEiedtopsep}{0pt}
\IEEEpeerreviewmaketitle

\newpage
%\vspace{2em} 

\section{Introduction}\label{sec:introduction}

Genome evolution is shaped not only by random mutations over generations but also horizontal genetic transfers between individuals of different species \cite{Chan_2013,marcussen2014ancient,Treangen_11_PLOS}. Phylogenetic trees have been used to study the evolution of life for over 200 years. The tree structure, however,  is not powerful enough to capture horizontal gene transfer,  hybridization and genetic recombination events, which occur frequently in viruses and bacteria. 
 This has led researchers to establish more general evolutionary models, including  phylogenetic networks and persistent homology, to investigate horizontal evolution \cite{Chan_2013,Doolittle,Moret_04_TCBB,Nakhleh_13_TREE}. A phylogenetic network is a rooted acyclic graph in which internal nodes are each of either indegree one and outdegree greater than one or 
outdegree one and indegree greater than one. The latter are called reticulations and are used to model reticulation events. The computational and mathematical aspects of phylogenetic networks
have been extensively studied over the past two decades (e.g. \cite{Gusfield_13_book,Huson_2010_book,Parida_JCB,Steel_16_Book,Wang_2001}).

On one hand, phylogenetic networks are very useful for dating
and inferring reticulation events \cite{Skog_Nature_15,Szohosi_2015}. On the other hand, it is extremely challenging to reconstruct
network models correctly from sequence data or from
gene trees \cite{huson2009computing,Yu_PNAS}.
Given that the phylogenies of numerous gene families and species have been
studied, phylogenetic networks are often reconstructed
and validated by examining their relationships with
existing gene trees and well-established clades \cite{Cardona_09_TCBB, Kanj_08_TCS}. The cluster containment problem (CCP) and the tree containment problem (TCP) thus arise  from study along these lines. The TCP asks whether or not a phylogenetic network displays a phylogenetic tree, whereas the CCP is about  determining whether or not a set of taxa is a cluster at a node in some phylogenetic tree displayed in a phylogenetic network.

The CCP and TCP are both NP-complete \cite{Kanj_08_TCS}, even for several special subclasses of binary phylogenetic networks \cite{Gambette_2016_BMB,van_Iersel_2010_IPL}.  Recently, good progress has been made on investigations into polynomial-time algorithms for solving the CCP and TCP on reticulation-visible networks \cite{Bordewich2016_AAM, gunawan2017decomposition, van_Iersel_2010_IPL}.  Most interestingly, the decomposition technique introduced in our work \cite{gambette2015solving, gunawan2017decomposition} and the visibility property lead to linear-time  algorithms for reticulation-visible  networks and nearly stable  networks \cite{gambette2015solving, Weller_17_ArXiv}, as well as fast algorithms for nonbinary networks \cite{Gunawan_16_ECCB,Lu_17_APBC}. Of note, the visibility property was first introduced by Huson and coworkers  to study phylogenetic networks \cite{Huson_2010_book}. An equivalence of this property appeared in a work of  
Lengauer and Tarjan on  flow analysis and program optimization\cite{Lengauer_79_ACM}.

In this paper, we develop fast methods for solving the CCP in arbitrary networks.    First, using the so-called small version of the CCP, namely determining whether a network displays a cluster at  a given tree node, we derive  a surprisingly  simple linear-time reduction from the CCP to the SAT problem. This implies an $O\left(t(n)2^{0.415n_r}\right)$ upper bound for the CCP, where 
$n_r$ is the number of reticulation nodes and $n$ is the number of nodes in the input network. This is mainly theoretically interesting, as $t(n_r)$ is a large sub-exponential function. We then present two practical methods. The first CCP algorithm is structurally very simple.  The second is designed through the small version of the CCP and bit more sophisticated. It has $O\left(n 2^{0.56 \psi}\right)$ time complexity, where $\psi$ is the number of invisible tree components in the given network. Our validation test shows that these two programs are much faster than what the worst-case time complexity suggests and also much faster than a method appearing in \cite{Lu_17_APBC}.

\section{Basic concepts and notations} \label{Sec:2}

\subsection{Phylogenetic networks}

\textit{Phylogenetic networks} over a set of taxa are acyclic rooted directed graphs in which the root is the unique node of indegree 0,  the nodes of outdegree 0 (called leaves) are labeled bijectively with taxa and each node is of either indegree one or outdegree one (Figure~\ref{Fig_CCP_example}). In this work, we simply call them networks. 

In a network, 
a node is a \textit{tree node} if it is of indegree one. Of note, leaves are tree nodes. For the sake of convenience, we add an incoming edge with an open end to the root to make it a tree node.
A non-root node is a \textit{reticulation} if it is of indegree greater than one and outdegree one. A reticulation is said to be at the {\it front}  if its child is a network leaf.  Note that the indegree of each reticulation might be more than two.  A network is \textit{binary} if every node is of total degree either three or one. 
For a network $N$, we use 
$\rho_N$ to  denote its root,
$\mathcal{V}(N)$ for its set of nodes, 
$\mathcal{E}(N)$ for its set of edges, 
$\mathcal{T}(N)$ for its set of nonleaf tree nodes (including the root), 
$\mathcal{R}(N)$ for its set of reticulations, and
$\mathcal{L}(N)$ for its set of leaves.

%An incoming edge of a reticulation (resp. a tree node or a leaf) is called a  \textit{reticulation edge} (resp. %\textit{tree edge}). A \textit{tree path} is a path consisting of only tree edges.

A node $u$ is a \textit{parent} of a node $v$, or $v$ is a child of $u$, if  $(u,v)\in \mathcal{E}(N)$. Two nodes are \textit{siblings} if they are the children of a common node. 
In general, a node $v$ is \textit{below} (or a \textit{descendant} of) a node $u$  if there is a path from $u$ to $v$.

A node $v$ is \textit{visible} on a leaf $\ell$ if every path from the root to $\ell$ contains $v$.
A node is visible if it is visible on some leaf, and is invisible otherwise. Note that front reticulations are visible and that the network root is visible on every leaf. 
%A network is \textit{reticulation-visible} if every reticulation in the network is visible 
%(Figure~\ref{Fig_CCP_example}A). 

The network $N[u]$ is defined as the network with the node set $\{v \in \mathcal{V}(N) : v \text{ is below } u\}$ and the edge set $\{(x,y)\in \mathcal{E}(N) : x \text{ and }y \text{ are below } u\}$. For a set of nodes $V$ and a set of edges $E$, $N - V - E$ is the digraph with the node set $\mathcal{V}(N)\setminus V$ and the edge set $\{(u,v) \in \mathcal{E}(N)\setminus E : u,v \notin V\}$. 
%If $V \cup E=\{x\}$,  we simply write $N - V - E$  as $N - x$.

\begin{figure}
    \centering
    \includegraphics[scale = 0.9]{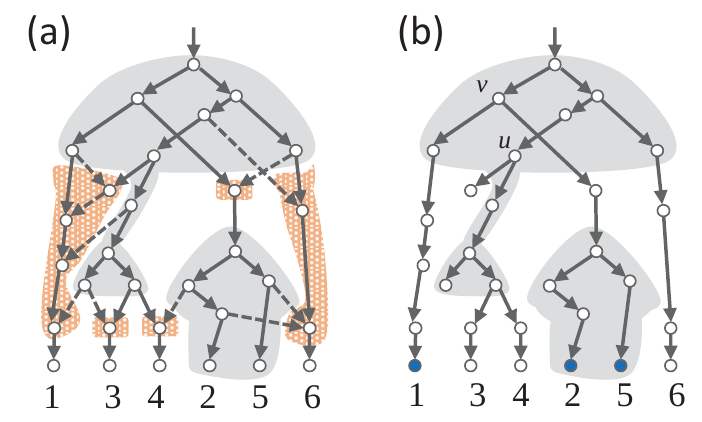}
    \caption{ (\textbf{a}) A network with six leaves with two nontrivial tree components (grey) and four trivial tree components $\{1\}, \{3\}, \{4\}, \{6\}$. There are  five reticulation components (orange and pattern fill).  For each reticulation, only one incoming arc is drawn as a solid arrow and the others as dashed arrows. 
    (\textbf{b}) Removing the dashed edges from the network in (a) leads to a spanning subtree of the network, in which tree node $u$ represents the cluster $\{3, 4\}$, whereas $v$ represents $\{1, 2, 5\}$.
}
    \label{Fig_CCP_example}
\end{figure}

\subsection{Tree components and visibility} \label{decompsection}

Let $N$ be a network.
Removing all reticulations from  $N$ yields a forest $N - \mathcal{R}(N)$ in which each connected component is a rooted subtree consisting of tree nodes only (Figure~\ref{Fig_CCP_example}a). 
Each connected component of the forest  is called a \textit{tree component} of  $N$  (see \cite{gunawan2017decomposition}). Similarly, removing all tree nodes from  $N$ yields a forest $N - \mathcal{T}(N)$ in which each connected component is an upside-down rooted subtree. 
 Each connected component of this forest  is called a \textit{reticulation component}, in which each directed edge points to the root of the component (Figure~\ref{Fig_CCP_example}a).

\begin{theorem}[Decomposition theorem, {\sc (\cite{gunawan2017decomposition})}]  \label{decomp}
Let $N$ be a network, and let  $K_0, K_1, \ldots, K_q$ be the tree components of $N$. In this case:
\begin{itemize}
\item[(1.)]  $\mathcal{T}(N) = \uplus_{i=0}^q \mathcal{V}(K_i)$, where $\uplus$ denotes the union of some disjoint sets.
  \item[(2.)]  The number of tree components is equal to one plus the number of reticulation components.
\end{itemize}
\end{theorem}

A tree component is \textit{trivial} if it is a component consisting of the leaf child of a front reticulation. It is \textit{exposed} if it is nontrivial and all the components below it, if any,  are trivial.

A component is \textit{visible} if its root is visible on a leaf; it is {\it invisible} otherwise. Of note, the component rooted at the network root is always visible. The number of invisible  tree components in a network is called the {\it invisibility number} of the network.

Let $v$ be  a reticulation in $N$.  It is {\it inner} if all of its parents are in the same tree component and  
{\it cross} otherwise.  A reticulation $v$ is {\it below} a tree component $K$ if the reticulation has a parent in $K$.
  A reticulation component $K'$ is {\it below} a tree component
$K$ if $K'$ contains a reticulation below $K$.

A network is said to be {\it reduced} if each reticulation component consists of a single reticulation, i.e., the child of every reticulation is a tree node.

\subsection{The cluster containment problem}

A subset of taxa is called a \textit{cluster}. For a tree $T$ and a node $v$ in $T$, the subset of the labels mapped to the leaves in the subtree $T[v]$ is said to be a cluster of $T$ (Figure~\ref{Fig_CCP_example}b).  
A network $N$ \textit{displays} a  cluster $B$ if there is a spanning subtree $T$ of $N$ such that $B$ is a cluster of $T$. 
The CCP for networks is formally defined:

\begin{quote}\smallskip
%{\bf The Cluster Containment Problem (CCP)}\\
\textit{Instance}: A network $N$ over $X$ and a subset $B$ of $X$.\\
\textit{Question}: Does $N$ display $B$? \smallskip
\end{quote}
In studies on the CCP, it is particularly convenient not to distinguish the leaves in $B$ (and those not in $B$) by simply coloring the leaves in $B$ with blue and the rest red. In this way, the CCP becomes to determine whether or not a network with leaves colored blue and red displays the set of blue leaves. We will work on this version in the rest of the paper.

\section{A simple CCP algorithm} 

\subsection{A reduction lemma}

 Let $K$ be an exposed tree component in $N$.  For a reticulations $s$ below $K$, we define: 
 \begin{eqnarray}
    E_{{in}}^{K}(s) &= \{ (u,s)\in \mathcal{E}(N) : u \in \mathcal{V}(K)  \}, \label{def_in_edges}\\
    E_{{out}}^{K}(s) &= \{ (u,s)\in \mathcal{E}(N) : u \notin \mathcal{V}(K) \}. \label{def_out_edges}
 \end{eqnarray}
 \noindent 
  We further use $R_{b}(K)$ and $R_{r}(K)$ to denote the set of  {\bf cross} front reticulations below $K$ whose children are blue or red leaves, respectively. That is:
\begin{eqnarray}
R_{b}(K)= \{ s\in {\cal R}(N): \mbox{ the child of $s$ is a blue leaf such that } E_{{in}}^{K}(s)\neq \emptyset \neq E_{{out}}^{K}(s)\};\\
R_{r}(K)= \{ s\in {\cal R}(N): \mbox{ the child of $s$ is a red leaf such that } E_{{in}}^{K}(s)\neq \emptyset \neq E_{{out}}^{K}(s)\}.
\end{eqnarray}
By merging all the blue leaves below into $K$ and sending all the red leaves away, we obtain the following subnetwork of $N$ (Figure~\ref{Fig3}):
 \begin{eqnarray}
   N^{K}_{{blue}} = N - \cup_{s\in R_b(K)} E_{{out}}^{K}(s) - \cup_{s\in R_r(K)} E_{{in}}^{K}(s). 
   \label{in_subnetwork}
 \end{eqnarray}
 Similarly, by merging all the red leaves below into $K$ and sending  all the blue leaves away, we obtain:
 \begin{eqnarray}
   N^{K}_{{red}} = N - \cup_{s\in R_r(K)} E_{{out}}^{K}(s) - \cup_{s\in R_b(K)} E_{{in}}^{K}(s). 
   \label{out_subnetwork}
 \end{eqnarray}

 \begin{figure}[b!]
    \centering
    \includegraphics[scale=0.9]{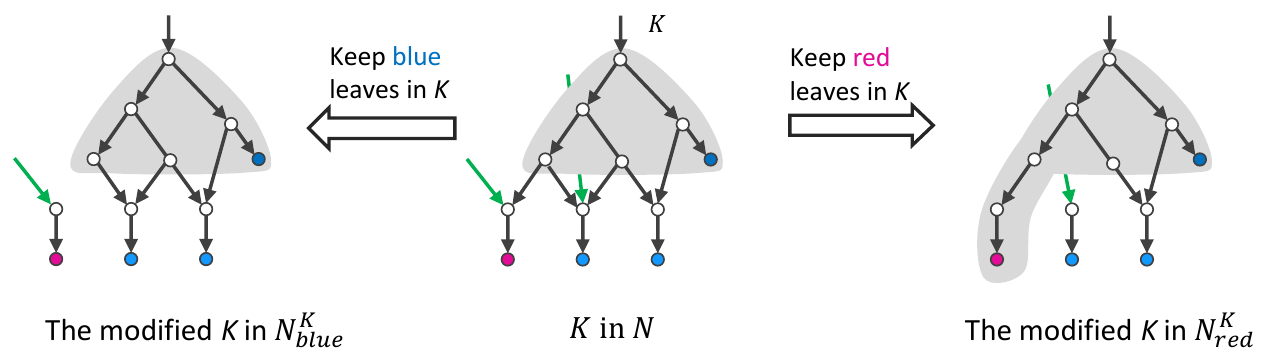}
    \caption{ Illustration of  $N^{K}_{{blue}}$ and $N^{K}_{{red}}$ as defined in Eqn. (\ref{in_subnetwork}) and (\ref{out_subnetwork}), respectively. Here,  green incoming edges are from other tree components  for the two front cross reticulations below $K$. Since $K$ is visible on a blue leaf in middle, its reticulation parent is still below $K$ after modification even in $N^{K}_{red}$.  }
    \label{Fig3}
 \end{figure}

\begin{lemma} \label{instancesplit}
Let $N$ be a network with leaves colored blue and red, and let $B$ denote the set of blue leaves. For an exposed tree component $K$,  $N$ then displays $B$ if and only if either $N^{K}_{{blue}}$  or $N^{K}_{red}$  displays it.
\end{lemma}
\begin{proof} ({\bf Necessity})  It is obvious, as  $N^{K}_{{blue}}$ and $N^{K}_{{red}}$ are both subnetworks of $N$.
  
({\bf Sufficiency})  We assume that $K$ is an exposed tree component in $N$ in which $B$ is displayed. Therefore,  a \textbf{spanning} tree $T$ of $N$ exists such that $B$ is the cluster at a node $v$ in $T$. 
For each cross reticulation $s$ below $K$, we select an edge $e_{{in}}(s)$ from $E_{{in}}^{K}(s)$ and an edge $e_{out}(s)$ from $E_{{out}}^{K}(s)$.
  Let $\rho_K$ be the root of $K$. Since $T$ spans $N$,   $T$ contains all the nodes of $N$. We consider the following cases.
  
  If  $v$ is above  $\rho_K$ in $T$, only blue leaves can be found below $v$. This implies that  
 none of the edges in $E_{{in}}^{K}(s)$ is in $T$ for each $s\in R_r(K)$. However, for
each cross reticulation $r \in R_b(K)$,  $T$ contains exactly an edge in either $E_{{in}}^{K}(s) $ or
$E_{{out}}^{K}(s) $ exclusively.   We define the subgraph $T'=(\mathcal{V}(N), \mathcal{E}(T'))$ of $N^{K}_{blue}$ as: 
  \begin{eqnarray}
   && \mathcal{E}(T') = \left(\mathcal{E}(T) -  \cup_{s\in R_b(K)} E_{{out}}^{K}(s) \right)
     \cup  \left\{ e_{{in}}(s) \;|\; s \in R_b(K) \text{ s.t. } \mathcal{E}(T) \cap E_{{out}}^{K}(s)
 \neq \emptyset  \right\},  \nonumber 
%\\
%   && E' = \left\{ e_{{in}}(s) \;|\; s \in R_b(K) \text{ s.t. } \mathcal{E}(T) \cap E_{{out}}^{K}(s)
 %\neq \emptyset  \right\} \cup 
 %   \left\{ e_{{out}}(s) \;|\; s \in R_r(K) \text{ s.t. } \mathcal{E}(T) \cap E_{{in}}^{K}(s) \neq 
%\emptyset  \right\}. \nonumber
  \end{eqnarray}
%
  %\textcolor{red}{$$\text{Old Eqtn: }\mathcal{E}(T')= \left\{\mathcal{E}(T)\backslash \left[\left( \cup_{s\in R_b(K)} E_{\mbox{out}}^{K}(s)\right)\cup \left(\cup_{s\in R_r(K)} E_{\mbox{in}}^{K}(s)\right)\right]\right\}\cup  \left[\left(\cup_{s\in R_b(K)} E_{\mbox{in}}^{K}(s)\right)\cup \left(\cup_{s\in R_r(K)} E_{\mbox{out}}^{K}(s)\right)\right].$$}
  %
  %
  In other words, $T'$ is obtained from $T$ by replacing any edge belonging to $ E_{{out}}^{K}(s)$
  with the selected edge entering $s$ from $K$ for each $s\in R_b(K)$.
  Since the indegree of each reticulation below $K$ remains one,    $T'$ is also a spanning subtree of $N^{K}_{{blue}}$ and $v$ has the same cluster in $T'$ as in $T$. This implies that  
$N^{K}_{{blue}}$ displays $B$ via $T'$.
  
  If $v$ is in $K$,  it is also possible for $T$ to contain an edge in $E_{{in}}^{K}(s) $ or
$E_{{out}}^{K}(s) $ for a reticulation $s \in R_r(K)$ if $s$ has a parent in $K$ but not below $v$. However, the incoming edge of a reticulation $s \in R_b(K)$ in $T$ must be an edge in $E_{{in}}^{K}(s)$ as $v$ is below $K$. In this case, $B$ is the cluster at $v$ in the spanning tree $T'$ of $N^{K}_{blue}$, defined as:
\begin{eqnarray}
   && \mathcal{E}(T') = \left(\mathcal{E}(T) -  \cup_{s\in R_r(K)} E_{{in}}^{K}(s)\right)
     \cup 
    \left\{ e_{{out}}(s) \;|\; s \in R_r(K) \text{ s.t. } \mathcal{E}(T) \cap E_{{in}}^{K}(s) \neq \emptyset  \right\}. \nonumber
  \end{eqnarray}

 If $v$ is incomparable with $\rho_K$, then no blue leaf will be found below $\rho_K$  in $T$. This implies that:
  $$\cup_{s\in R_b(K)} E_{\mbox{in}}^{K}(s) \cap \mathcal{E}(T) =\emptyset. $$
  
 %{$$\text{Old Eqtn:  }\cup_{s\in R_b(K)} E_{\mbox{in}}^{K}(s) \cap \mathcal{E}(T) =\emptyset, \; 
 %\cup_{s\in R_r(K)} E_{\mbox{in}}^{K}(s) \subseteq \mathcal{E}(T).$$
 
 In this case, we define  $T''=(\mathcal{V}(N), \mathcal{E}(T''))$, where:
 \begin{eqnarray}
 \label{Eqn_T''}
\mathcal{E}(T'')= \left(\mathcal{E}(T) - \cup_{s\in R_r(K)} E_{{out}}^{K}(s)\right)
 \cup  
 \left\{ e_{{in}}(s) \;|\; s \in R_r(K) \text{ s.t. } \mathcal{E}(T) \cap E_{{out}}^{K}(s) \neq \emptyset  \right\}.  \nonumber 
\end{eqnarray}
Hence, $T''$ is obtained from $T$ by replacing any edge belonging to $E_{{out}}^{K}(s)$
  with the selected edge entering $s$ from $K$ for each $s\in R_r(K)$.  Note that  $T''$ is a spanning tree of $N^{K}_{{red}}$. This implies that  $N^{K}_{{red}}$ displays $B$ via $T''$.
\end{proof}

  \subsection{A simple but effective approach}

Let $N$ be a network with the leaves being colored blue and red, which is not necessarily binary.
By Lemma 1, determining whether $N$ displays the cluster $B$ consisting of blue leaves can be done by
(1) selecting an exposed tree component $K$ and (2) recursively determining whether  $B$ is displayed by either $N^{K}_{blue}$ or $N^{K}_{red}$.  This simple algorithm is illustrated in Figure~\ref{Fig4}. It turns out to be asymptotically better than the first non-trivial CCP algorithm  developed by us in \cite{Lu_17_APBC}, as showed by our analysis in Appendix A.

Note that if  the selected tree component $K$ is not visible on red (resp. blue) leaves, only blue (resp. red) leaves are found in $\bar{K}$   in  $N^{K}_{blue}$ (resp. $N^{K}_{red}$). 
 Let $r$ be the reticulation parent of $\rho(K)$. 
 If $\bar{K}$ contains only blue  (resp. red) leaves,  we will replace $\bar{K}$ with a blue (resp. red)  leaf  so that $r$ will become a front reticulation with a blue (resp. red) leaf child  in $N^{K}_{blue}$ (resp. $N^{K}_{red}$) in the next step.  If $\bar{K}$ does not contain any colored leaves, it will be removed, together with $r$ and the edges entering $r$. 
The resulting network will be denoted as $\mbox{Update}(N^{K}_{blue})$
(resp. 
$\mbox{Update}(N^{K}_{red})$). With this notation, we present our first CCP algorithm  in Appendix A, where its time complexity is also analyzed.

\begin{figure}[b!]
    \centering
    \includegraphics[scale=0.9]{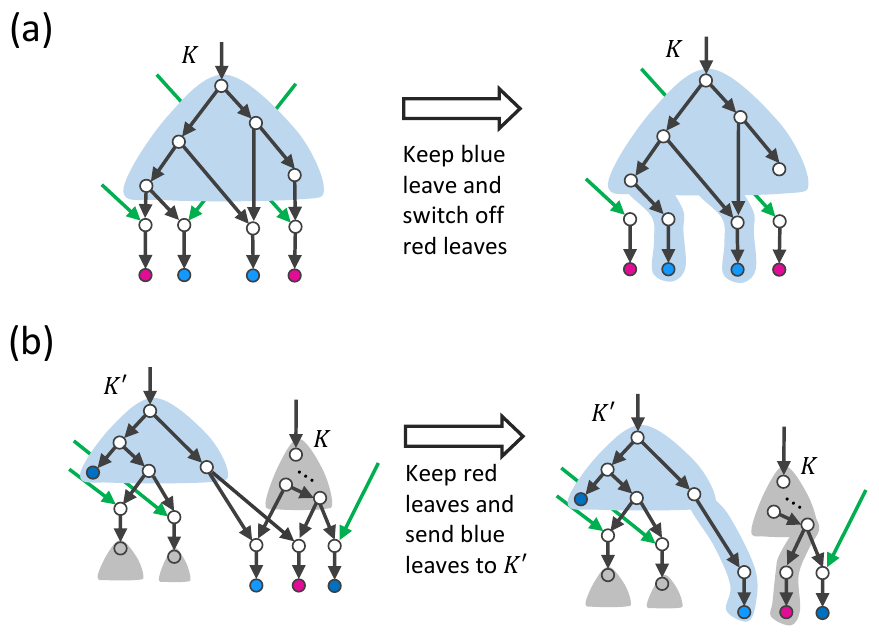}
    \caption{Two possible cases on 
     an exposed tree component $K$ for which only simplification is needed. 
     ({\bf a}) $K$ is visible.  It is simplified by switching off the red leaves if $K$ is visible on blue and switching off the blue leaves otherwise. 
     ({\bf b}) $K$ is invisible but adjacent to a tree  component $K'$ whose color status is known.
     In this figure, $K'$ is assumed to be visible on blue leaves. $K$ is then simplified by sending each network leaf $\ell$
     under both tree components to $K'$ if $\ell$ is blue and keeping it otherwise.
    }
\label{Figure06}
\end{figure} 

\section{Reduction to The SAT Problem}
\label{reduction_sec}

 We now consider the following simple version of the CCP.

\begin{quote}\smallskip
 {\bf The Small Cluster Containment Problem} (SCCP)\\
 {\it Instance}: A network $N$ with colored leaves $X$ and a tree node $v\in {\cal  T}(N)\backslash{\rho(N)}$.\\
 {\it Question}: Does $N$ display the cluster $B$ of blue leaves at $v$?
\end{quote}\smallskip

In the instance part of the SCCP, we assume that $v$ is a tree node that is different from ${\rho(N)}$, the network root.  The reason for this is because (i) $X$ is the only cluster displayed  at ${\rho(N)}$ in $N$ and  
(ii) a cluster is displayed at a reticulation node if and only if it is displayed at its child.  We further assume that $v$ is the root of a tree component in $N$. Otherwise, if $v$ is in a tree component $K$ but not $\rho(K)$, we may construct a network $N'$ from $N$ by adding a tree node $t$ in the edge entering $\rho(K)$, a reticulation node $r$ in the edge entering $v$ and an edge $(t, r)$. It is easy to see that $N$ and $N'$ display the same set of soft clusters at $v$.  These two assumptions lead to an efficient reduction to the SAT problem and another better CCP algorithm.

Assume that $N$ displays blue leaves at $v$ via a spanning tree $T_b$. 
Since $T_b$ spans $N$, it contains all the nodes in $N$ and is obtained by removal of  all but one edge entering $r$ from each $r\in {\cal R}(N)$ (Figure~\ref{Fig_CCP_example}b).  Each tree component $K$, therefore,  becomes a subtree of $T_b$. This implies that  it is either entirely below $v$ or disjoint from the subtree consisting of vertices below $v$ in $T_b$ under the assumption that $v$ is the root of a component.  

 Let $C_{T_b}(x)$ denote the cluster at a vertex $x$ in $T_b$. 
If a tree component $K$ is below $v$ in $T_b$, $C_{T_b}(\rho (K))$  comprises  only blue leaves if it is nonempty. Otherwise, $C_{T_b}(\rho (K))$ comprises red leaves if it is nonempty.  Therefore, 
$T_b$  induces a blue or red coloring of  all the tree components with the following properties:

%\begin{quote}
\begin{itemize}
  \item The tree component $K_v$ rooted at $v$ must be blue and so are the tree components containing blue leaves.
   \item The tree component $K_{\rho(N)}$ containing $\rho(N)$,  those containing red leaves or not below $v$ must be red.
   \item Let $K$ be a tree component and let the parent of $\rho(K)$ be $r_K\in {\cal R}(N)$ such that $K_{\rho(N)}\neq K\neq K_{v}$.   If $K$ is blue,   $T_b$ contains  an edge $(u, r_k)$, where $u$ is a tree-node  in another blue tree component $K'$ above $K$.  Similarly, $T_b$ contains an edge from a tree node in a red tree component to $r_K$ if $K$ is red.  
\end{itemize}
Conversely, this coloring is sufficient to demonstrate that $N$ displays blue leaves at $v$.
These observations lead to a linear reduction from the SCCP to the SAT problem, which is formally proved in Appendix B.

\begin{theorem}
\label{SAT_reduction}
The SCCP on reduced  networks can be linearly reduced to the SAT problem in such a way  that a SCCP  instance $Q: (N, v)$  is transformed into a SAT instance $S_Q$ which contains 
$|\chi(N)|$ variables and the number of terms in each clause is bounded by one plus the maximum indegree of a reticulation node in $N$, where $\chi(N)$ is the number of tree components in $N$. 
\end{theorem}

Theorem~\ref{SAT_reduction} has two implications.  In theory, since an instance of the SCCP involving a binary network is transformed into a 3-SAT instance, 
the SCCP and CCP have an upper bound $O\left(t(n)1.3334^{|{\cal R}(N)|}\right)$, which is about $O\left(t(n)2^{0.415|{\cal R}(N)|}\right)$,  using the best derandomization of Sch\"{o}ning's algorithm \cite{Moser_11_STOC}, where $t(n)$ is a sub-exponential function of the number $n$ of nodes in the input binary network.  It improves upon a known upper  bound $O\left(n2^{0.5|{\cal R}(N)|}\right)$ \cite{Kanj_08_TCS}. 
In practice, we can implement a computer program for both SCCP and CCP by just calling a popular SAT solver.

\section{A SCCP Algorithm}

In this section, we will work on reduced networks. For an arbitrary network, we may replace each reticulation component by a non-binary reticulation node such that the resulting network displays a cluster if and only if the original network displays it. 

\subsection{Technical lemmas}

Our SCCP algorithm is  a recursive procedure that colors tree components in a bottom-up manner to search for a coloring satisfying the properties listed in Section~\ref{reduction_sec}. In each step, we simplify the current network by either (i)  coloring at least one exposed tree component,  removing reticulation edges from  this newly colored one to front reticulations  with a different color and  contracting an empty tree component if any,  or (ii) recursively calling the procedure  on each coloring extension if there are two possible coloring extensions. For the sake of presentation,  we  define the following concepts and notation:
\begin{itemize}
 %\item A reticulation is said {\it trivial} if its child is a leaf.
 
\item Two tree components $K'$ and $K''$ are {\it adjacent} (via a front reticulation) if there is a front reticulation $r$ such that $r$ has a parent in each of the tree components. For example, $K_1$ and $K_2$ are adjacent in $N$ appearing in Figure~\ref{Fig4}.  Two adjacent tree components are said to be {\it neighbors} of each other.

\item Only front reticulations are found below an exposed tree component. A leaf is said to be below 
 an exposed tree component $K$, if it is the child of a front reticulation below $K$.  

\item
In each step, we may color a nontrivial tree component blue (resp. red) in the sense that only blue (resp. red) leaves are allowed to be added into this component. 
 Once  a red (resp. blue) leaf is detected in a blue (resp. red) tree component,  the current coloring extension will stop.  Furthermore, after a leaf is merged into an internal tree component $K$ with unknown color status,  $K$ will be colored with the leaf's color at the same time.   The color of a tree component will never be changed once it is assigned. Importantly, an internal colored tree component does not necessarily contain a leaf with that color when its color is assigned during the execution of the algorithm.

\item If an exposed tree component below $v$ is visible on a blue leaf,  it will be colored blue by the algorithm.   Similarly, an exposed tree component visible on a red leaf will be colored red by the algorithm.  Therefore,  the color status can be considered to be a generalization of the visibility concept for the tree components in this study.

\item Assume that some tree components are colored in $N$, in which we try to determine whether $B$ is displayed at $v$ or not. This partial coloring is {\it good} if $N$ displays $B$ at $v$ via a spanning tree in which  blue components are below $v$ and red ones are not below $v$, whereas uncolored tree components may or may not be below $v$.  
\vspace{0.5em}
\end{itemize}

\begin{figure}[t!]
    \centering
    \includegraphics[scale=1]{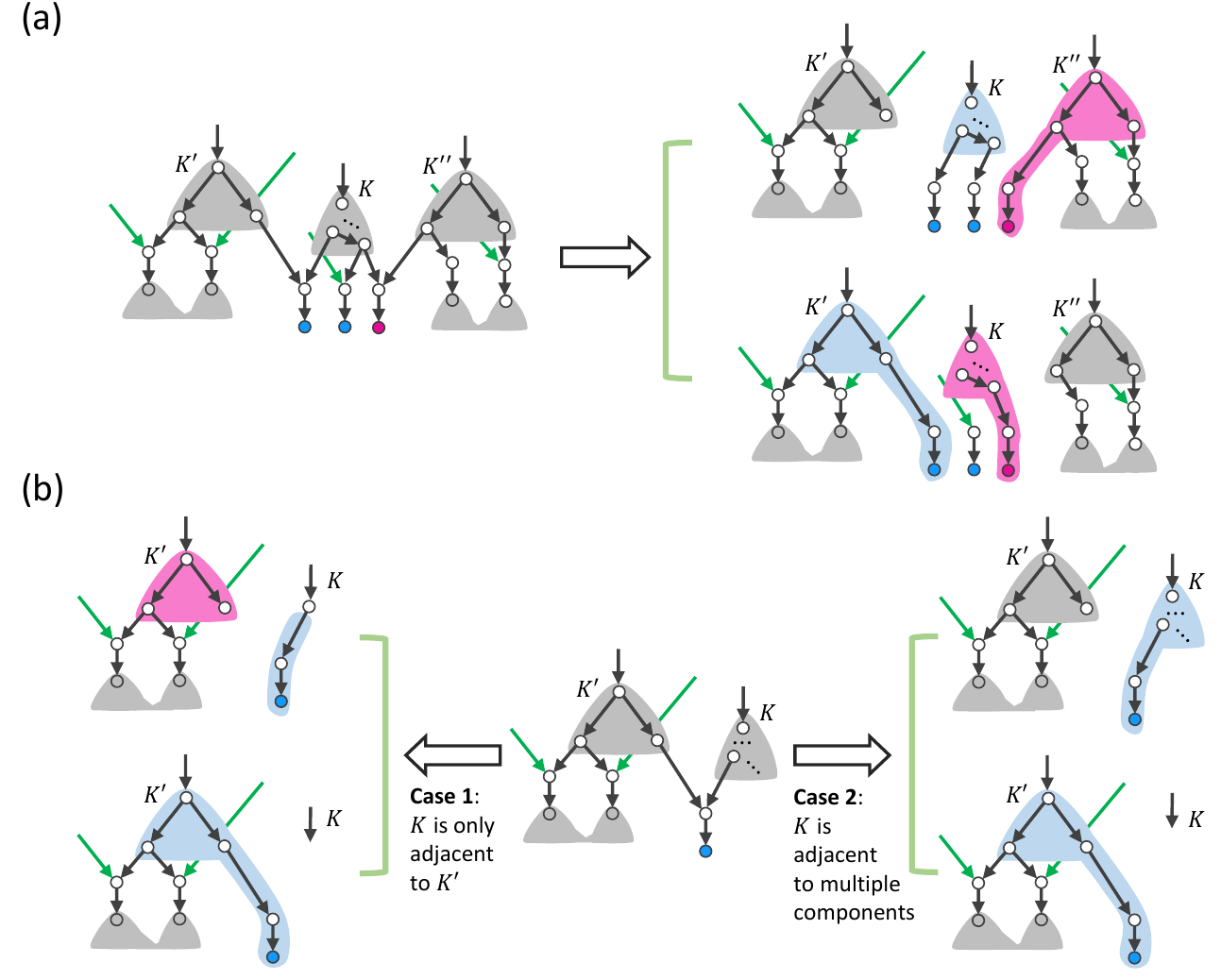}
    \caption{Two possible cases on 
     an invisible exposed tree component $K$ for which two recursive calls are needed. 
     ({\bf a}.) There are blue and red leaves below $K$. (1) Simplify $K$ by keeping blue leaves and switching off red leaves and contract $K$ into a blue leaf before making a recursive call. (2) Simplify $K$ by keeping red leaves and switching off blue leave and contract $K$ into a red leaf before making a recursive call. Here, $K$ may be adjacent only to one invisible tree component.   ({\bf b.}) The leaves below $K$ are of the same color. Assume that blue leaves are below $K$ as shown here.  We consider the two cases separately. If $K$ is adjacent to only one invisible tree component $K'$,  (1) set $K'$ to be red  and contract $K$ into a blue leaf before making a recursive call; and (2) set $K'$ to be blue and switch all the blue leaves to $K'$ before removing $K$ and making a recursive call.
     If $K$ is adjacent to two invisible tree components or more, 
     (1) keep the blue leaves in $K$ before contracting $K$ into a blue leaf and making a recursive call; and (2) switch off all blue leaves before removing $K$ and making a recursive call. 
     }
    \label{Figure07}
\end{figure}

Our SCCP algorithm is based on the following lemmas. 
Let $K$ be a selected  exposed tree component for color extension in the network $M$ derived from the input network when a step of the algorithm starts. We consider whether or not $K$ is visible and whether or not the neighbors of $K$ are colored in order to color $K$ and further simplify $M$. Let $B$ be the set of blue leaves in $M$.

\begin{lemma} \label{visible} Let $K$ be a visible and  exposed  tree  component  in $M$ (Figure~\ref{Figure06}a).  
    \begin{itemize}
    \item[(1.)] If $M$ displays blue leaves at $v$, $K$ cannot be visible on both blue and red leaves.
   \item[(2.)]  Define  $M'$ to be        
    $M^{K}_{{blue}}$  (Eqn.~(\ref{in_subnetwork})) if $K$ is visible on blue leaves and  $M^{K}_{{red}}$ (Eqn.~(\ref{out_subnetwork})) if $K$ is visible on red leaves.  $M$ displays $B$ at $v$ if and only if $M'$ displays $B$ at $v$.  
     \end{itemize}
\end{lemma}
%Lemma 4
\begin{proof} It is proved in Appendix C.
%    (1.)  Let $K$ be visible on a blue leaf $\ell'$ and a red leaf $\ell''$.  Suppose, on the contrary, that $M$
% displays  $B$ at $v$ via a spanning tree $T$.  Since $T$ is a spanning tree,  two paths
%$P'$ and $P''$ exist from $\rho(M)$ to $\ell'$ and $\ell''$ in $T$, respectively. Since $K$ is visible on $\ell'$ and $%\ell''$,  $P'$  and $P''$ must  both contain $\rho(K)$. If $K$ is below $v$ in $T$, then $\rho(K)$ must be below $v$ and $\ell''$ thus belongs to $B$. If $K$ is not below $v$ in $T$,  it is incomparable to $v$, as $K$ is exposed.  In %this case, we  find that  $\ell'$ does not belong to $B$. This is a contradiction. 
%
%(2.)  The sufficiency  is clear, as $M'$ is a subnetwork of $M$.   Let $K$ be visible on a blue leaf $\ell$. Since $K$ is exposed, $\ell$ is either in $K$ or the child of an inner front reticulation,  whose parents are all in $K$ (Figure~%\ref{Figure06}a).  Assume that $M$  displays $B$ at  $v$ via a spanning tree $T$. Thus $\ell$ is below $v$ and %thus the unique  path from $v$ to $\ell$ must contain $\rho(K)$ in $T$, as every grandparent of $\ell$ is in $K$. %Therefore $K$ must be blue.  Similarly, $K$ must not be below $v$  and thus $K$ must be  red if $K$ is visible
% on a red leaf. 
\end{proof}

Lemma~\ref{visible} suggests that an exposed tree  component that is visible on blue (resp. red) leaves has to be colored blue (resp. red) if not colored.

Let $K$ be an exposed  tree component that is invisible  and let $K'$ be a tree component  adjacent to $K$ in $M$. Recall that $R_{b}(K)$ (resp.
$R_{r}(K)$) is the set of cross  front reticulations such that they are each below $K$ and their unique child is a blue (resp. red) leaf. Define:
$$R_{blue}(K, K')=R_b(K)\cap R_b(K'); \;\;R_{red}(K, K') = R_r(K)\cap R_r(K').$$   We also use $M(K_{red}, K'_{blue})$ to denote the spanning subnetwork of $M$ with the following edge set:
\begin{eqnarray}
 {\cal E}(M) - \cup_{s\in R_{blue}(K, K')}  E^{K}_{in} (s) - \cup_{s\in R_{red}(K, K')} E^{K'}_{in} (s). 
\end{eqnarray}
In other words, $M(K_{red}, K'_{blue})$ is obtained from $M$ by removal of all the edges $(u, s)$ such that $u\in K$ for each $s\in R_{blue}(K, K')$ and all the edges $(v, s)$ such that $v\in K'$ for each $s\in R_{red}(K, K')$. 
It is true that  the red leaves below front reticulations in $R_{red}(K, K')$ are merged  into $K$, whereas  the blue leaves in $R_{blue}(K, K')$ are merged into $K'$ in $M(K_{red}, K'_{blue})$,  as shown in 
Figure~\ref{Figure06}b.
Symmetrically, $M(K_{blue}, K'_{red})$ can be defined.  

\begin{lemma} \label{reduction}
    Let $K$ be an exposed tree component that is not colored but adjacent to a colored tree component $K'$ in $M$ (Figure~\ref{Figure06}b). This partial coloring can then be extended in either (i) $M(K_{blue}, K'_{red})$ if $K'$ is red or  (ii)  $M(K_{red}, K'_{blue})$ if $K'$ is blue.
\end{lemma}
%Lemma 6
\begin{proof}  The proof is similar to that of Lemma~\ref{instancesplit}.  \end{proof}

Lemma~\ref{reduction} suggests the fact that when an exposed tree component $K$ is colored,  we can simplify $K$ (and thus the network) by merging  a blue  (resp. red) leaf $\ell$ below $K$ into a blue (resp. red) neighbor $K'$ if $\ell$'s has a grand parent in $K'$. Of note, a similar fact can be found in \cite{Kanj_08_TCS}.

We now consider the inconsistency between the color of a leaf commonly below an uncolored $K$ and 
the neighbors of $K$ by defining  the following parameters:
  \begin{itemize}
 \item $I_b(K)$: The set of {\bf red} neighbors $K'$ of $K$ such that 
$R_{blue}(K, K')\neq \emptyset$.  
\item $I_r(K)$: The set of {\bf blue} neighbors $K'$ of $K$ such that 
$R_{red}(K, K')\neq \emptyset$. 
\item $S_{b}(K)$: the  set of uncolored neighbors  $K'$ of $K$ such that 
$R_{blue}(K, K')\neq \emptyset$. 
\item $S_{r}(K)$: the  set of uncolored neighbors  $K''$ of $K$ such that 
$R_{red}(K, K'')\neq \emptyset$. 
\end{itemize}
Note that at least one of these sets is nonempty for an invisible tree component $K$. 

\begin{lemma} \label{compatibility_lma}
Let $K$ be an uncolored, exposed tree component in $M$ where the tree components are partially colored, and let $I_b(K)$ and $I_r(K)$ be defined as above. Denote this partial tree component coloring as ${\cal C}$.  Then, we have the following facts.
\begin{itemize}
\item[(1.)]  If  $I_b(K)$ is nonempty, but $I_r(K)$ is empty,  ${\cal C}$ can only be extended by coloring $K$ blue.  
\item[(2.)] If $I_r(K)$ is nonempty, but $I_b(K)$ is empty,  ${\cal C}$  can only be extended by coloring $K$ red. 
\item[(3.)]  If  $I_b(K)$ and $I_r(K)$ are both nonempty, ${\cal C}$  is not good. 
\end{itemize}
\end{lemma} 
%lemma5
\begin{proof}
    Assume that $I_b(K)$ is nonempty. Let $K' \in I_b(K)$, and let $s \in R_{blue}(K, K')$. Then, $K'$ is colored red and $s$ has a blue leaf as its child.  Thus, the blue child of $s$ can only be merged into $K$ and $K$ is colored blue.   Similarly, if $I_{r}(K)$ is nonempty, then $K$ can only be colored by red.  if both $I_r(K)$ and $I_b(K)$ are nonempty,  the coloring is not good.
\end{proof}

\begin{lemma} \label{assignment_lma} 
Let $K$ be an uncolored, invisible and exposed tree component in $M$ where tree components are partially colored such that
  $I_b(K)=I_r(K)=\emptyset$ and let $S_b(K)$ and $S_r(K)$ be defined as above.  Denote this partial coloring as ${\cal C}$.  Then, the following facts are true only if $M$ is a {\bf binary} network. 
\begin{itemize}
\item[(1.)]   If  $|S_b(K)|+|S_r(K)|\geq 2$, ${\cal C}$ can be extended by either (i) coloring $K$ blue and the components in $S_r(K)$ red or (ii) coloring $K$ red and 
the components in $S_b(K)$ blue (Figure~\ref{Figure07}a and Case 2 in Figure~\ref{Figure07}b). 
\item[(2.)]   If  $S_b(K) =\{K'\}$ and $S_r(K)=\emptyset$, ${\cal C}$ can be extended by (i) coloring $K'$ blue and switching all blues below $K$ to $K'$ to make $K$ empty  or (ii) coloring $K'$ red and $K$ blue (Case 1 in Figure~\ref{Figure07}b). 
\item[(3)] If  $S_b(K)=\emptyset$ and $S_r(K)=\{K'\}$, ${\cal C}$ can be extended by (i) coloring $K'$ red and switching all red below $K$ to $K'$ to make $K$ empty  or (ii) coloring $K'$ blue and $K$ red. 
\end{itemize}
\end{lemma}
%lemma 6
\begin{proof} It is given in Appendix C. \end{proof}

 %   (1) Assume that $|S_b(K)| + |S_r(K)| \geq 2$. 
 %   Note that we can color $K$ with either blue or red. 
%    If $K$ is colored blue, then for every tree component $K'$ in $S_r(K)$, we have that $K'$ is uncolored and  %$R_{red}(K', K') \neq \emptyset$. Then,  we have to merge the red leaf below each front reticulation in $R_{red}%(K', K') $ to $K'$ and thus $K'$ must be colored red. Similarly, if $K$ is colored red,  we must color every tree %component in $S_b(K)$ blue.
    
  %  (2) Assume that $S_b(K) = \{K'\}$ and $S_r(K) = \emptyset$.  $K'$ can be colored either red or blue. 
 %As $K$ is invisible, $K$ does not contain any network leaf. Moreover, every reticulation below $K$ must be a cross front reticulation, and it must be in $R_b(K) \cup R_r(K)$. As $K'$ is the only neighbor of $K$ and $K' \in S_b(K)$, this means that every reticulation below $K$ is in $R_{blue}(K, K')$ and $K$ has no red leaf below it.  If $K'$ is colored blue,   we can remove the edges in $E_{in}^K(s)$ for every reticulation $s$ below $K$, which makes $K$ has no labeled leaf below it. If $K'$ is colored red, then we must color $K$ blue to get the blue leaves below $K$ and $K'$ .  Hence (2) holds.
    
  %  The proof for (3) is analogous to that of (2). This completes the proof.
%\end{proof}

\subsection{Algorithm and time analysis}

Taken together, Lemmas 2 to 5 imply the SCCP algorithm in Appendix C.

\begin{theorem}
\label{thm3}
The SCCP can be solved in $O\left(n \cdot 2^{0.552\psi(N)}\right)$ for binary networks $N$ with $n$ nodes, where $\psi(N)$ is the number of invisible tree components of $N$.
\end{theorem}
%theorem 3
\begin{proof}
 After pre-processing the network and topologically sorting all the tree components in $N$, both Step 1 and Step 2 take constant time for each recursive call. 
    
   For each exposed tree component $K$, there are at most $|{\cal V}(K)|+1$  edges leaving $K$,  as each tree node is of outdegree 2. This implies that at most $|K|+1$ front reticulations are found below $K$ and $K$ are adjacent to at most $|K|+1$ tree components. Since $K$ is exposed,  the visibility of $K$ can be done by checking whether there is an inner reticulation below $K$, which can be in $O(|{\cal V}(K)|)$ if we keep updating the status of being cross for each of the reticulations in each step.  Therefore, each of the "if" conditions can be checked using at most $O(|K|)$ many basic set-arithmetic operations.
    
    In each of Step 3.1, 4.2, 5.1, 6.1, 6.5, 7.1, and 7.4, we need to compute and simplify $N_{blue}^K$, $N_{red}^K$, $N(K_{blue}, K'_{red})$, $N(K_{red}, K'_{blue})$, or $N - \mathcal{V}(K)$. In Steps 6.2 and 6.6, we check the nodes in every neighbor of $K$. In total, these steps take $O(|\mathcal{V}(N)|)$ basic set-arithmetic operations.

 Let $\psi(N)$ be the number of tree components that are uncolored in $N$ when the algorithm starts. If Step 3, 6, or 7 is executed when working on a uncolored component $K$, then $K$ is either removed or becomes colored by the end of the step.  Between two consecutive executions of  Step 6 or Step 7,  the algorithm may execute a series of Steps 3,  4 and 5. 
    Since different executions consider different tree components in Steps 3 and 4 and different neighbors in Step 5, 
    the algorithm uses at most $O(|\mathcal{T}(N)|)$ basic set-arithmetic operations between two consecutive executions of Steps 6 and 7 in total. Each execution of these two steps can create two recursive calls on two simplified networks.

   In summary, let $f(k)$
   denote the worst-case time taken by the algorithm for binary networks with $k$ uncolored tree components. Note that the color of  a  tree component remains unchanged once it is assigned.  Assume $K$ is the selected tree component at a step.  If leaves of both colors are found below $K$, then Step 6 is executed. At the end of Step 6,  $K$ and at least one neighbor of $K$ become colored in both $N^{K}_{blue}$ and $N^{K}_{red}$, so $\psi \left(N^{K}_{blue}\right)\leq \psi (N)-2$ and  $\psi \left(N^{K}_{red}\right)\leq \psi (N)-2$.

    Similarly, if $K$ only has a neighbor $K'$ and all the leaves below $K$ are of the same color, then
Step 7 is executed.  In this step,  $K'$ is colored, whereas $K$ is either removed or is colored with the other color in both $N^{K}_{blue}$ and $N^{K}_{red}$.  Again, we have $\psi \left(N^{K}_{blue}\right)\leq \psi(N)-2$ and 
    $\psi\left(N^{K}_{red}\right)\leq \psi^ (N)-2$.

    Finally, if $K$ has two or more  neighbors and only leaves of a color $c$ are found below $K$,  Step 6 is then executed.  Assume $c$ is blue.   $K$ is colored blue in $N_{blue}^K$.  However, $K$ is removed and
at least two of its neighbor tree components are colored blue in $N^{K}_{red}$. 
For this case, therefore,  we have $\psi (N^{K}_{blue})\leq \psi (N)-1$ and 
    $\psi(N^{K}_{red})\leq \psi (N)-3$.
    
   Taken together, the above facts imply:  
    $$f(k)\leq \max\{2f(k-2), f(k-1) + f(k-3)\} +c|\mathcal{V}(N)|$$
    for a small constant $c$. By induction,  we can show that:
    $$f(k) \leq  c'  |\mathcal{V}(N)| \cdot 2^{\gamma \psi(N)},$$
    where $\gamma \approx 0.55416$ is the real solution of $2^{-x} + 2^{-3x} = 1$ and $c'$ is a constant. This completes the proof.
\end{proof}

\noindent {\bf Remarks} (a)  The SCCP algorithm presented in Appendix C is only correct for binary networks, as  the statements in Lemma~\ref{assignment_lma} do not hold for arbitrary networks. As such, the time complexity given in Theorem~\ref{thm3} is only valid for binary networks. For arbitrary networks, the coloring settings for the neighbors of the selected component in Step 6 and Step 7 need to be revised; and the time complexity of the corresponding algorithm can be expressed as $O\left(|{\cal V}(N)|2^{0.5 h(N)}\right)$, where 
$h(N)$ is equal to the sum of the super-indegrees of the reticulation components minus the number of the reticulation components. Here, the super-indegree of a reticulation component $R$ is defined to be the number of the tree components $K$ such that there exists an edge from a node in $K$ to a reticulation in $R$.  

(b)  We implement the second CCP program by simply calling the SCCP algorithm presented here on each tree node.   

\begin{figure}[t!]
    \centering
    \includegraphics[scale=1]{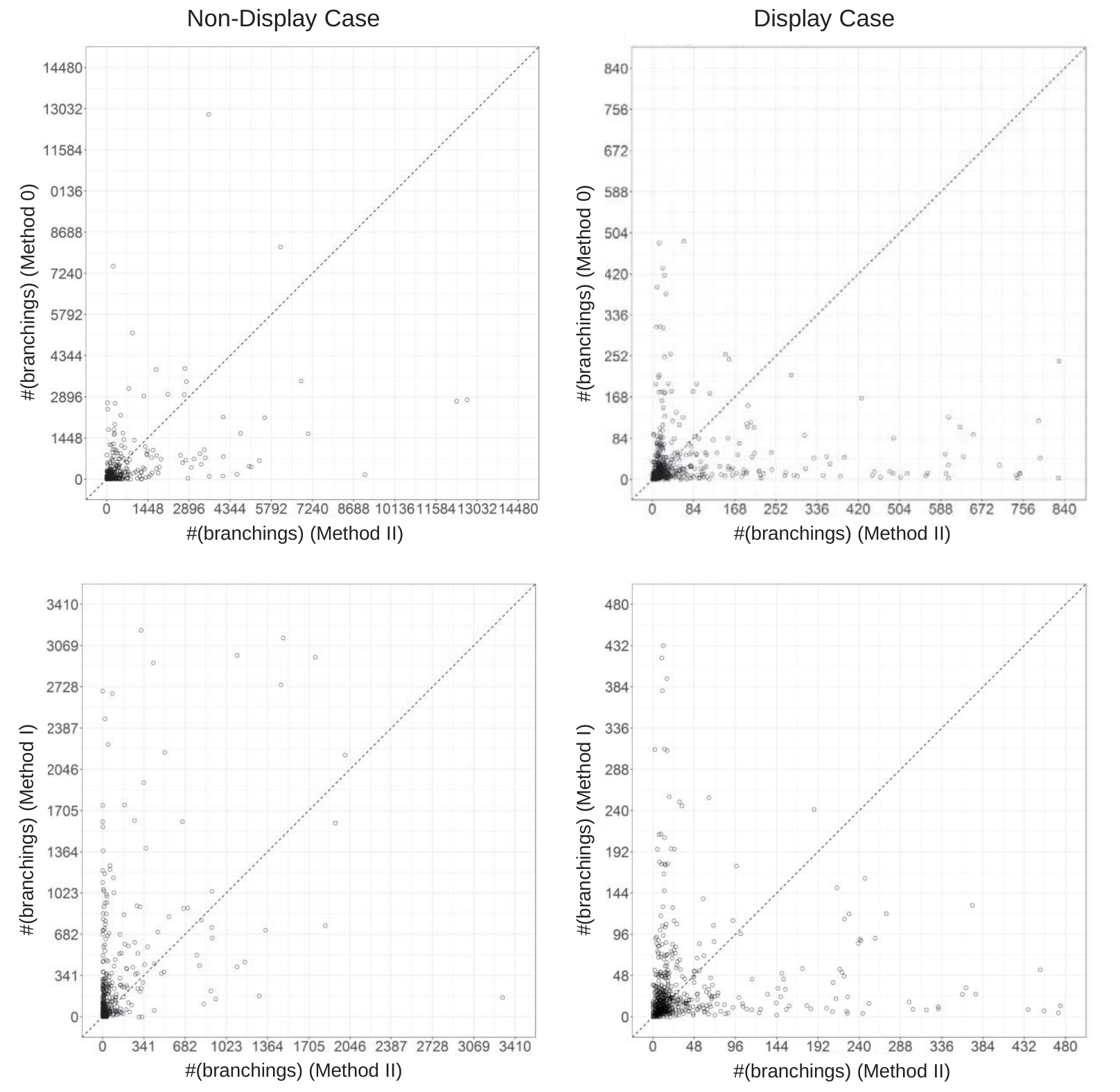}
    \caption{Performance comparison of the three methods on random networks in Group C. Method 0:  the method in \cite{Lu_17_APBC}; Method 1: the simple method presented in Section 3; 
Method 2, the method presented in section 5 (see remark in the section). \label{Fig08}}
\end{figure}

\section{Validation tests}

To validate  two new methods, we implemented and compared them, as well as each of them against 
the program reported in \cite{Lu_17_APBC},  on 3,000 random  binary networks on a PC cluster with 32 GB RAM and 24 cores.  The old method, the methods presented in Sections 3 and 5 are called Method 0, I and II, respectively. 
The run time for each method largely depends on the number of steps in which two possible colorings of a tree component are separately examined. This {\it branching} number reflected by the exponential term in their time complexity was used to examine the performances of the three methods.

The random networks were generated by a computer program available online (http://phylnet.univ-mlv.fr/tools/randomNtkGenerator.php). The 3,000 random network-cluster pairs were divided into three equal groups:
\begin{quote}
  {\bf Group A}.   Each network had six leaves and 20 reticulations; each cluster had three or four leaves.

  {\bf Group B}.  Each network had eight leaves and 30 reticulations; each cluster had four or five leaves.

  {\bf Group C}.  Each network had ten leaves and 40 reticulations; each cluster had five or six leaves.. 
\end{quote} 
The performances of the three methods on networks in these three groups showed similar trends. Hence, 
only the analyses on Group C are summarized in Figure~\ref{Fig08}. The analyses on Group A and Group B can be found in Supplementary Figures~\ref{FigS2} and~\ref{FigS3}.  These validation tests suggests the following facts:
%\begin{quote}

  1. For each method and each network, the  branching number in a non-display case in which the given cluster was not displayed is much larger than in a display case. As such, we analyzed the display  and non-display  cases in different figures. In general, the branching number in a non-display case is  at least ten times as large as that in a display case.\vspace{0.5em}

 2. Methods 0 and II had similar performance in non-display cases. Their performance were much worse than Method I. The left bottom panel in Figure~\ref{Fig08} shows that the branching number for Method I is more than 20 to 30 times as fast as the other two in many non-display cases. \vspace{0.5em}

3. Methods I and II had comparable performance in display cases. Method I performed much better than Method II in hundreds of display cases and vice versa.  Both performed significantly better than Method 0. \vspace{0.5em}
%\end{quote}

The reason for the bad performance of Method II in a non-display case is that it determines the display of a cluster node by node. Since for each node,  its branching number is small, estimated as $2^{0.552 \psi(N)}$ (Theorem~\ref{thm3}),  it is expected to surpass Method I when the input network gets larger and larger.

We also discover a simple reduction from the SCCP to the SAT problem. This allows us to develop new practical SCCP and CCP programs through plugging in a popular SAT solver such as the MiniSAT (http://minisat.se/).  
Since the reduction was just discovered, we have yet to compare our method with any  SAT-solver based method. We will do it in the final version.

\section{Conclusions} 

The decomposition technique 
and visibility property \cite {Huson_2010_book} were shown to be powerful for solving the TCP
on arbitrary networks \cite{gambette2015solving,gunawan2017decomposition,Gunawan_16_ECCB}.  In this work, applying these two 
techniques, we have developed two fast algorithms for
solving the CCP.  These two algorithms were implemented
in C.   Our validation test shows that the methods presented in Sections 3 and 5 have different advantages. The former is better in non-display cases, whereas the latter is better in display cases.

The programs have been extended to programs for computing 
the Soft Robinson-Foulds distance for networks.  Both programs facilitate reconstructing and
 validating networks in comparative genomics.

\ifCLASSOPTIONcompsoc
  \section*{Acknowledgments}
\else
  \section*{Acknowledgment}
\fi

LX Zhang would like to thank Dominik Scheder for discussions on the SAT problems.
This work was financially supported by the Singapore Ministry of Education Academic Research Fund Tier 1 grant [R-146-000-238-114].

%\bibliographystyle{elsarticle-num}
%\bibliography{sample}

\begin{thebibliography}{10}
\expandafter\ifx\csname url\endcsname\relax
  \def\url#1{\texttt{#1}}\fi
\expandafter\ifx\csname urlprefix\endcsname\relax\def\urlprefix{URL }\fi
\expandafter\ifx\csname href\endcsname\relax
  \def\href#1#2{#2} \def\path#1{#1}\fi

\bibitem{Bordewich2016_AAM}
M.~Bordewich, C.~Semple, Reticulation-visible networks. Adv. Appl. Math. 78
  (2016) 114--141.

\bibitem{Cardona_09_TCBB}
G.~Cardona, M.~Llabr{\'e}s, F.~Rossell{\'o}, G. Valiente, 
Metrics for phylogenetic networks I: Generalizations of the Robinson-Foulds metric. IEEE-ACM Trans. Comput. Biol. Bioinform.  6(2009), 46--61.

\bibitem{Chan_2013}
J.~M. Chan, G.~Carlsson, R.~Rabadan, Topology of viral evolution.  Proc. Natl.
  Acad. Sci. 110~(46) (2013) 18566--18571.

\bibitem{Doolittle}
W.~F.~Doolittle,   Phylogenetic classification and the universal tree. {Science} 248(1999),  2124--2128.


\bibitem{marcussen2014ancient}
T.~Marcussen, et~al., Ancient
  hybridizations among the ancestral genomes of bread wheat. Sci. 345~(6194)
  (2014) 1250092.

\bibitem{Kanj_08_TCS}
I.~A. Kanj, L.~Nakhleh, C.~Than, G.~Xia, Seeing the trees and their branches in
  the network is hard. Theor. Comput. Sci. 401 (2008) 153--164.

\bibitem{van_Iersel_2010_IPL}
L.~van Iersel, C.~Semple, M.~Steel, Locating a tree in a phylogenetic network.
  Inf. Process. Lett. 110~(23) (2010) 1037--1043.


%\bibitem{RECOMB2015}
%P.~Gambette, A.~Gunawan, A.~Labarre, S.~Vialette, L.~X.~Zhang, {Locating a tree in
%  a phylogenetic network in quadratic time}, in: Proceedings of the 19th Annual
%  International Conference on Research in Computational Molecular Biology
 % ({RECOMB 2015}), Vol. 9029 of LNCS, 2015, pp. 96--107.

\bibitem{gambette2015solving}
P.~Gambette, A.~Gunawan, A.~Labarre, S.~Vialette, L.~X.~Zhang, Solving the tree containment problem in linear time for nearly stable phylogenetic networks. Discrete Applied Math., in press. Also partially presented on RECOMB'2015.

\bibitem{gunawan2017decomposition}
A.~Gunawan, B.~DasGupta, L.~X.~Zhang, A decomposition theorem and two algorithms
  for reticulation-visible networks. Inf. Comput. 252 (2017) 161--175. Also presented on RECOMB'2016.

\bibitem{Gambette_2016_BMB}
P.~Gambette,  L.~van Iersel, S.~Kelk, F.~Pardi, C.~ Scornavacca. Do branch lengths help to locate a tree in a phylogenetic network?  Bull. Math. Biol. 78 (2016): 1773-1795.

\bibitem{Gusfield_13_book}
D.~Gusfield,  ReCombinatorics: the algorithmics of ancestral recombination graphs and explicit phylogenetic networks. MIT Press, Mass., USA, 2014.

\bibitem{Huson_2010_book}
D.~H.~Huson, R.~Rupp, C.~Scornavacca. Phylogenetic networks: concepts, algorithms and applications. Cambridge University Press, Cambridge, UK, 2010. 

\bibitem{huson2009computing}
D.~H.~Huson, R.~Rupp, V.~Berry, P.~Gambette, C.~Paul, Computing galled networks
  from real data. Bioinform. 25 (2009) i85--i93.

\bibitem{Gunawan_16_ECCB}
A.~Gunawan, B.~Lu, L.~X.~Zhang, A program for verification of phylogenetic network
  models. Bioinform. 32 (2016) i503--i510. Also presented on ECCB'2016.

\bibitem{Lengauer_79_ACM}
T.~Lengauer, R.~E.~ Tarjan,  A fast algorithm for finding dominators in a 
flowgraph. ACM Trans. Prog. Lang. Sys. 1 (1979), 121--141.


\bibitem{Lu_17_APBC}
B.~Lu, L.X.~Zhang, H.~W. Leong, A program to compute the soft Robinson--Foulds
  distance between phylogenetic networks. BMC genomics 18 (2017) 111.

\bibitem{Moret_04_TCBB}
 B.~M.~E.~Moret,  L.~Nakhleh,  T.~Warnow, \textit{et~al}.  
   Phylogenetic networks: Modeling, reconstructibility, and
  accuracy.  {IEEE-ACM Trans. Comput. Biol. Bioinform.}, { 1} (2004),   13--23.

\bibitem{Moser_11_STOC}
R.~A.~Moser,  D.~Scheder, A full derandomization of Sch\"{o}ning's k-SAT algorithm.
 In Proc. 43rd Annual ACM STOC, pp. 245-252. ACM, 2011.


\bibitem{Nakhleh_13_TREE}
L. Nakhleh, Computational approaches to species phylogeny inference and gene
  tree reconciliation. {Trends Ecol. Evol.}  {28} (2013),   719--728.

\bibitem{pardi2015}
F.~Pardi, C.~Scornavacca, Reconstructible phylogenetic networks: do not
  distinguish the indistinguishable, PLoS Comput. Biol. 11~(4) (2015)
  e1004135.



\bibitem{Parida_JCB}
L.~Parida,  Ancestral recombinations graph: a reconstructability perspective using
random-graphs framework.  {J. Comput. Biol.} 17 (2010), 1345--1370.

\bibitem{Skog_Nature_15}
P. Skoglund, S.~Mallick, M.~Bortolini, N.~Chennagiri, T.~H{\"u}nemeier, M.~Petzl-Erler, F.~Salzano, N.~Patterson, D.~Reich. Genetic evidence for two founding populations of the Americas. Nature 525 (2015), 104--108.

\bibitem{Steel_16_Book}
M.~Steel,  Phylogeny: Discrete and random processes in evolution. SIAM, Philadelphia, USA, 2016.

\bibitem{Szohosi_2015}
G.~J~Sz{\"o}ll{\H{o}}si, A.~Dav{\'\i}n, E.~Tannier, V.~Daubin, B.~ Boussau,
Genome-scale phylogenetic analysis finds extensive gene transfer among fungi. Phil. Trans. R. Soc. B 370 (2015), 20140335.

\bibitem{Treangen_11_PLOS}
T.~J.~Treangen,  E. Rocha, Horizontal transfer, not duplication, drives the expansion of protein families in prokaryotes. PLoS Genetics 7(2011): e1001284.

\bibitem{Wang_2001}
L.~Wang,  K.~Zhang,  L.~X.~Zhang,   Perfect phylogenetic networks with
  recombination. {J. Comp. Biol.} {8} (2001), 69--78.


\bibitem{Weller_17_ArXiv}
M.~Weller,  Linear-time tree containment in phylogenetic networks. arXiv preprint arXiv:1702.06364, 2017.

\bibitem{Yu_PNAS}
Y.~Yu, J.~Dong, J., K.~J.~Liu,  \textit{et~al}. Maximum likelihood inference of reticulate evolutionary histories. {Proc. Natl. Acad. Sci. U.S.A.} {111} (2014), 16448--16453.



\end{thebibliography}
 \newcommand{\noop}[1]{}

\clearpage

\section*{Appendix A:  The first CCP Algorithm and its time complexity }\vspace{1em}

The pseudo-code of the CCP algorithm presented in Section 3 is given below. \vspace{1em}

{
%\begin{figure}[h]
\small
\begin{tabular}{l}
 \hline
\\
\hspace*{18em} {\sc CCP Algorithm One} \vspace{0.5em}
\\
/* {\bf Input}: $N$, a  reduced network in which leaves are colored blue and red, and $\beta$, the number of blue leaves; */ \\
/* {\bf Output}: "TRUE" if $N$ displays the cluster consisting of blue leaves and "FALSE" otherwise.  */ \vspace{0.5em}\\
/* Update (Network N'):  a procedure that contracts a simplified tree component $K'$ in $N'$ into a leaf */ \\
/*~~~~or removes $K'$ completely depending whether $K'$ contains a colored leaf or not. */ \vspace{0.5em}\\
{\sc IS-B-Display}(Network $N$, int $\beta$)\{\\ \\
~1. If ($\beta == 1$ ) return TRUE; \vspace{0.5em}\\
~2. Compute an exposed tree component $K$;
\vspace{0.5em}\\
~3. If (no colored leaf is found below $K$) \{\\
\hspace*{2em} 3.1.~~ update $N$ by removing $K$;\\
\hspace*{2em} 3.2.~~ return {\sc IS-B-Display} (Update($N$), $\beta$);\\
~~~~\}\vspace{0.5em}\\
~4. If ($K$ is visible only on blue leaves) \{\\
\hspace*{2em} 4.1.~~ construct $N^{K}_{blue}$;\\ \hspace*{2em} ~~~~~~~~$\beta=\beta-$~\#(blue leaves below $K$ in $N$)+1;\\
\hspace*{2em} 4.3.~~ return {\sc IS-B-Display} (Update($N^{K}_{blue}$), $\beta$); \\
~~~~\}\vspace{0.5em}\\
~5. If ($K$ is visible on red leaves) \{\\
\hspace*{2em} 5.1.~~ construct $N^{K}_{blue}$;\\ 
\hspace*{2em} 5.2.~~ return "TRUE" if blue leaves are displayed at a node in the updated $K$ in  $N^{K}_{blue}$;\\
\hspace*{2em} 5.3.~~ if ($K$ is also visible on blue leaves) \{ \\
\hspace*{6em} ~~~~~~~ return "FALSE"; \\
\hspace*{3em}~~~~~~\} else \{\\ \hspace*{6em} 5.5.~~ construct $N^{K}_{red}$;\\
\hspace*{6em} 5.6.~~ {\sc IS-B-Display} (Update($N^{K}_{red}$), $\beta$);\\
\hspace*{3em}~~~~~ \}\\ 
~~~~\}\vspace{0.5em} \\
%
%~6. If ($K$ is visible on both red and blue leaves) \{\\
%\hspace*{2em} 6.1.~~ construct $N^{K}_{blue}$;\\ 
%\hspace*{2em} 6.2.~~ return whether blue leaves are displayed at a node in the updated $K$ in  $N^{K}_{blue}$;\\
%~~~~\}\vspace{0.5em} \\
%
~6. If ($K$ is invisible) \{ \\
\hspace*{2em} 6.1.~~ construct $N^{K}_{blue}$;\\
\hspace*{2em} ~~~~~~~~$\alpha=\beta -$~\#(blue leaves below $K$ in $N$)+1;\\
\hspace*{2em} 6.3.~~ $X=$~{\sc IS-B-Display}(Update($N^{K}_{blue}$), $\alpha$);\\
\hspace*{2em} ~~~~~~~~if ($X$~==~TRUE) return TRUE; \\
\hspace*{2em} 6.5.~~ construct $N^{K}_{red}$;\\
\hspace*{2em} 6.6.~~ return {\sc IS-B-Display}(Update($N^{K}_{red}$), $\beta$);\\
~~~~\}\vspace{0.5em}\\
\}\vspace{0.5em}\\
\hline     
\end{tabular}
%\caption{A CCP algorithm derived from Lemmas 1-3.}
%\label{CCP1}
%\end{figure}
}
\vspace{2em}

It is not hard to see that the proof of Lemma~\ref{instancesplit} implies the following facts. 
\vspace{0.5em}

 \noindent {\bf Lemma A.1}  Let $K$ be a visible, exposed tree component of $N$.  
 \begin{itemize}
  \item[(1)] If $K$ is visible only on blue leaves, then $N$ displays $B$ if and only if $N^{K}_{{blue}}$ displays $B$.
  \item[(2)] If $K$ is visible only on red leaves, then  $N$  displays $B$ if and only if either $N^{K}_{{red}}$ displays $B$ or $N^{K}_{{blue}}$  displays $B$ at a node in $K$.
  \item[(3)] If $K$ is visible on both red and blue leaves, then $N$ displays $B$ if and only if  $N^{K}_{{blue}}$ displays $B$ at a node in $K$.
 \end{itemize}
\vspace{2em}

\noindent {\bf Theorem A.2}
The CCP algorithm presented in Appendix A takes $O\left( |\mathcal{R}(N)| \cdot 2^{\psi(N)} + |\mathcal{E}(N)|\right)$ time on arbitrary reduced networks $N$, where $\psi (N)$ denotes the {\it invisibility}  number of $N$, the number of invisible tree components.
%\engin{theorem}
%\noindent {\bf Proof of Theorem 2}.
\begin{proof}
%The algorithm is clearly correct. We now examine its time complexity.  
Step 1 takes constant time. Step 2 takes $O(|\mathcal{E}(N)|$ for the input network $N$. One simple way of finding an exposed component is to compute the non-trivial components in $N$ and then topologically sorting them \cite{gunawan2017decomposition}.  In each of the following recursive steps,  an exposed component can be taken from the end of this sorted list of components  in constant time.  

Let $|K|$ denote the number of tree nodes in $K$. Since each tree node is of outdegree two, there are at most $|K|+1$ reticulations having parents in $K$. Therefore,  the ``if" condition in each of Steps 4 to 6 can be verified in $|K|$ basic set-arithmetic operations. 

In Step 3.1, the removal of $K$ can be done by deleting the edges entering the parent of $\rho(K)$. In each of Steps 4.1. 5.1, 5.5, 6.1 and 6.5, we need only to remove certain edges entering the front reticulations below $K$.  In total,  these steps take 
at most {$|\mathcal{R}(N)|$} basic set-arithmetic operations. 

In Step 3, we remove $K$, as there is no network leaf below $K$.
The property that no leaf is below $K$ implies that for any network leaf in $N$,  there is a path from $\rho_N$ to this leaf that does not go through $K$. This implies that $K$ is an invisible component in $N$. Therefore,  Step 3 and Step 6 are executed only when invisible exposed components are examined. 

Between two consecutive executions of  Step 6,  the algorithm may execute a series of Steps 3, 4 and 5. Since only certain edges entering the front reticulations below $K$ are removed in  Step 3.1, 4.1, or 5.1 and different executions remove different components and thus {different} edges, 
the algorithm uses at most $|\mathcal{R}(N)|$ basic set-arithmetic operations between two consecutive executions of Step 6  in total.

In summary, let {$f(k)$}
denotes the worst-case time taken by the algorithm for networks with $k$ invisible components. Since visible components remain visible in both  $N^{K}_{blue}$ and $N^{K}_{red}$, $\psi(N^{K}_{blue})\leq \psi (N)-1$ and 
$\psi(N^{K}_{red})\leq \psi (N)-1$ if $K$ is an invisible component. Therefore, 
$$f(k)\leq 2f(k-1)+c|\mathcal{R}(N)|,\;\; 
|\mathcal{R}(N^{K}_{blue})|\leq |\mathcal{R}(N)|,\;\;
|\mathcal{R}(N^{K}_{red})|\leq |\mathcal{R}(N)|,$$
for some small constant $c$. This implies that 
for  $N$, 
the algorithm takes at most $O\left(|\mathcal{R}(N)|\cdot 2^{\psi(N)} + |\mathcal{E}(N)| \right)$.
\end{proof}

\section*{Appendix B.  Proof of Theorem 2}

\noindent {\bf Theorem 2.} The SCCP on networks can be linearly reduced to the SAT problem in such a way  that a SCCP  instance $Q: (N, v)$  is transformed into a SAT instance $S_Q$ which contains 
$|\Gamma(N)|$ variables and the number of terms in each clause is bounded by the maximum indegree of a reticulation node in $N$. 
\begin{proof}

  Consider a SCCP instance $Q$ consisting of a network $N$
with colored leaves in which no tree component contains both blue and red leaves and a node $v$, the root of a tree component $K_v$. 
Without loss of generality, we may assume that $N$ is reduced, i.e. each reticulation node consisting of a single reticulation node \cite{pardi2015}. 
   A tree component is said to be {\it reachable}  from $v$ if $v$ is an ancestor of  the nodes in this component. We denote $\Gamma (N)$ the set of all tree components in $N$ and by 
$\Gamma(v)$ the set of tree components reachable from $v$. We also denote
by ${\cal R}(v)$ the set of reticulations reachable from $v$.

  For each tree component $K$, we introduce a variable $v(K)$. For each reticulation $s\in {\cal R}(v)$, let $K_s$ be the tree component rooted at the child of $s$ and let  $K^{(s)}_1, K^{(s)}_2, \cdots, K^{(s)}_m$ be all the components containing a parent of $s$, we introduce the following two clauses:
\begin{eqnarray}
  C_{s1}:  v\left(K_s\right)\lor \overline{v\left(K^{(s)}_1\right)} \lor  \overline{v\left(K^{(s)}_2\right)} \lor \cdots \lor  \overline{v\left(K^{(s)}_1\right)};
\nonumber\\
 C_{s2}:   \overline{v\left(K_s\right)}\lor {v\left(K^{(s)}_1\right)} \lor  {v\left(K^{(s)}_2\right)} \lor \cdots \lor  {v\left(K^{(s)}_1\right)}. \label{type1_clause}
\end{eqnarray}  
Additionally,  the variables corresponding tree components containing a blue leaf  and $K_v$ will each take the value "TRUE",  Hence, for each of these components $K$, we introduce a single term clause: 
\begin{eqnarray} A_K:   v(K),  K\in \left\{K_v, \{\ell\} : \ell\in B\right\}. \label{type2_clause}\end{eqnarray}
Similarly, the variables corresponding tree components containing a red leaf  and those not reachable from $v$ will each take the value "FALSE", for which we introduce the following  single term clauses:
\begin{eqnarray}B_K:  \overline{v(K)}, \;\;\; K=\{\ell\},  \ell\in {\cal L}(N)\backslash B, \mbox{ or } K\in \Gamma(N)\backslash \Gamma(v). \label{type3_clause} \end{eqnarray}

We now show that a truth assignment satisfies the SAT  instance   $S_Q$ consisting of the clauses listed in Eqn.~(\ref{type1_clause})-(\ref{type3_clause})  if and only if $N$ displays all the blue leaves at $v$.
 %$$ \land_{s\in R(v)} \left (C_{s1} \land C_{s2}\right). $$ 

Assume that $N$ displays $B$ at $v$ via a spanning $T$. Then, the variable corresponding to a tree component is assigned ``TRUE" if  this component is below $T$.  The clauses in Eqn.~(\ref{type2_clause})-(\ref{type3_clause}) are clearly satisfiable.  Consider a  reticulation $s$ in $N$. We let $p_s$ and $c_s$ be the parent and child of $s$ in $T$, respectively, which are unique, and let $P_T(\rho(N), s)$ be the path from $\rho(N)$ to $s$ in $T$. If $P_T(\rho(N), s)$ contains $v$, the tree components containing $p_s$ and $c_s$ are below $v$ and thus the corresponding variables take the value ``TRUE", this makes the two clauses defined in  Eqn.~(\ref{type1_clause}) for $s$ satisfiable. If $P_T(\rho(N), s)$ does not contain $v$, the tree components containing $p_s$ and $c_s$ are not below $T$ and thus the corresponding variables take the value ``FALSE". This also makes the two clauses defined  in  Eqn.~(\ref{type1_clause}) for $s$ satisfiable.

Conversely, assume there is a truth assignment that makes $N_Q$ satisfiable. Consider a tree component $K$. If $v(K)$ takes the value ``TRUE",
there must be a path from $v$ to $\rho(K)$ that passes through the tree components $K'$ whose corresponding variables take the value
"TRUE".  Assume that it does not hold.  There must be a reticulation $s$ such that the tree component rooted at its child is ``TRUE" but all the tree components containing at least one parent of $s$ are ``FALSE",  implying that the clause $C_{s2}$ defined in Eqn.~(\ref{type1_clause}) is false. This is a contradiction. 

Since each tree component containing a blue leaf takes the value ``TRUE",  the subnetwork consisting of tree nodes in all ``TRUE" components and reticulation nodes between them has  a spanning subtree that roots  at $v$ and contains all blue leaves but not red ones. This shows that $N$ displays all blue leaves at $v$. 
\end{proof}
%
%
%
%
%\newpage

\section*{Appendix C:  The SCCP Algorithm and its time complexity}\vspace{2em}

\noindent {\bf Proof of Lemma 2}

   (1.)  Let $K$ be visible on a blue leaf $\ell'$ and a red leaf $\ell''$.  Suppose, on the contrary, that $M$ displays  $B$ at $v$ via a spanning tree $T$.  Since $T$ is a spanning tree,  two paths
$P'$ and $P''$ exist from $\rho(M)$ to $\ell'$ and $\ell''$ in $T$, respectively. Since $K$ is visible on $\ell'$ and $\ell''$,  $P'$  and $P''$ must  both contain $\rho(K)$. If $K$ is below $v$ in $T$, then $\rho(K)$ must be below $v$ and $\ell''$ thus belongs to $B$. If $K$ is not below $v$ in $T$,  it is incomparable to $v$, as $K$ is exposed.  In this case, we  find that  $\ell'$ does not belong to $B$. This is a contradiction. 

(2.)  The sufficiency  is clear, as $M'$ is a subnetwork of $M$.   Let $K$ be visible on a blue leaf $\ell$. Since $K$ is exposed, $\ell$ is either in $K$ or the child of an inner front reticulation,  whose parents are all in $K$ (Figure~\ref{Figure06}a).  Assume that $M$  displays $B$ at  $v$ via a spanning tree $T$. Thus $\ell$ is below $v$ and thus the unique  path from $v$ to $\ell$ must contain $\rho(K)$ in $T$, as every grandparent of $\ell$ is in $K$. Therefore $K$ must be blue.  Similarly, $K$ must not be below $v$  and thus $K$ must be  red if $K$ is visible on a red leaf. 

The completes the proof of Lemma 2.
\vspace{1em}

\noindent {\bf Proof of Lemma 5} 

%\begin{proof}
    (1) Assume that $|S_b(K)| + |S_r(K)| \geq 2$. 
    Note that we can color $K$ with either blue or red. 
    If $K$ is colored blue, then for every tree component $K'$ in $S_r(K)$, we have that $K'$ is uncolored and  $R_{red}(K', K') \neq \emptyset$. Then,  we have to merge the red leaf below each front reticulation in $R_{red}(K', K') $ to $K'$ and thus $K'$ must be colored red. Similarly, if $K$ is colored red,  we must color every tree component in $S_b(K)$ blue.
    
    (2) Assume that $S_b(K) = \{K'\}$ and $S_r(K) = \emptyset$.  $K'$ can be colored either red or blue. 
 As $K$ is invisible, $K$ does not contain any network leaf. Moreover, every reticulation below $K$ must be a cross front reticulation, and it must be in $R_b(K) \cup R_r(K)$. As $K'$ is the only neighbor of $K$ and $K' \in S_b(K)$, this means that every reticulation below $K$ is in $R_{blue}(K, K')$ and $K$ has no red leaf below it.  If $K'$ is colored blue,   we can remove the edges in $E_{in}^K(s)$ for every reticulation $s$ below $K$, which makes $K$ has no labeled leaf below it. If $K'$ is colored red, then we must color $K$ blue to get the blue leaves below $K$ and $K'$ .  Hence (2) holds.
    
    The proof for (3) is analogous to that of (2). This completes the proof.
%\end{proof}

The pseudo-code of the SCCP algorithm developed in in Section 5.2 is presented below. \vspace{1em}

%\begin{figure}[ht]
%\centering
{\small
\begin{tabular}{l}
 \hline \\
\hspace*{18em} {\sc  SCCP Algorithm} \vspace{1em}
\\
/*~~~~{\bf Input}: $N$,  a binary network with blue and red leaves, and $v$,  the root of a tree component in $N$; */ \\
/* {\bf Output}: "TRUE" if $N$ displays the cluster consisting of  blue leaves at $v$ and "FALSE" otherwise. */\vspace{0.5em}\\
/* $c(K)$: the color status for tree component $K$. Its value can be "blue", "red", "U" (unassigned). */
\vspace{1em}\\
{\sc IS-B-Display} (network $N$, node $v$) \{\vspace{0.5em}\\ 
~~1. Set the component $C_v$ containing $v$ to be blue and  other tree components not below $C_v$ red; \vspace{0.5em}\\
~~2. Select an exposed tree component $K$ in $N$;
\vspace{0.5em}\\
~~3. If ($K$ has no leaf below it) \{\\
\hspace*{2em} 3.1.~~ update $N$ by removing $K$;\\
\hspace*{2em} 3.2.~~ {\sc IS-B-Display} (Update($N$), $v$);\\
~~~~~\}\vspace{0.5em} \\
~~4. If (~($K$ is visible  on leaves of color $c$) or ($c(K) \neq\;U$)~) \{\\
\hspace*{2em} 4.1.~~   if ($K$ is visible on leaves of two colors or $c\neq c(K)$) 
~ return ``FALSE";\\
\hspace*{2em} 4.2.~~   if ($c(K)==\mbox{``blue"}$ \&\& $c==\mbox{blue}$) \{\\ \hspace*{5em} if ( ($K$ contains $v$) \& (all blue leaves are below $K$) return "TRUE"; else  $M :=N^{K}_{{blue}}$;\\
\hspace*{4em}~~ \} else $M :=N^{K}_{{red}}$;\\ 
\hspace*{2em} 4.3.~~ {\sc IS-B-Display}$\left(\mbox{Update}\left(M\right), v\right)$; \\ 
~~~~~\}\vspace{0.5em} \\
~~~~/* {\bf in Steps 5\textendash7, $K$ is invisible and its color status is not assigned.} */ \\
~~5. If ($K$ is adjacent to a tree component $K'$ s.t. $c(K')=$ ``blue" or ``red")) \{\\
\hspace*{2em} 5.1.~~   if ($c(K')==\mbox{``blue"}$) $M:=N(K'_{blue}, K_{red})$; else $M:=N(K'_{red}, K_{blue})$;\\ 
\hspace*{2em} 5.2.~~ {\sc IS-B-Display}
$\left(\mbox{Update}(M), v\right)$;\\ 
~~~~~\}\vspace{0.5em} \\
~~6. If (~(leaves of two colors are   found below $K$) or ($K$ is adjacent to 2+ tree components without color))  \{ \\
\hspace*{2em} 6.1.~~ compute $M:=N^{K}_{{blue}}$;\\
\hspace*{2em} 6.2.~~ for each neighbor $K'$ of $K$, $c(K') = $ "red"  if $K'$ contains red leaves in $M$;\\
\hspace*{2em} 6.3.~~ $X=$~{\sc IS-B-Display}(Updated($M$), $v$);\\
\hspace*{2em} ~~~~~~~~if ($X$~==~TRUE) return "TRUE"; \\
\hspace*{2em} 6.5.~~ compute $M := N^K_{red}$;\\
\hspace*{2em} 6.6.~~ for each neighbor $K'$ of $K$, $c(K') = $ "blue" if $K'$ contains blue leaves in $M$;\\
\hspace*{2em} 6.7.~~  return {\sc IS-B-Display}($M$), $v$);\\
~~~~~\}\vspace{0.5em}\\ 
~~7. If (only  leaves of a color $c$ are below $K$)~\&\&~($K$ has only a neighbor $K'$ without color status))  \{ \\
\hspace*{2em} 7.1.~~ if ($c==$ ``blue") \{  $M:=N^{K}_{{blue}}$; ~
$c(K')$:= "red"; \} else \{ $M:=N^{K}_{{red}}$; ~
  $c(K')$ := "blue"; \} \\
\hspace*{2em} 7.2.~~ $X=$~{\sc IS-B-Display}(Updated($M$), $v$);\\
\hspace*{2em} ~~~~~~~~if ($X$~==~TRUE) return TRUE; \\
\hspace*{2em} 7.4.~~ $M:= N - \mathcal{V}(K)$;\\
\hspace*{2em} ~~~~~~~~if ($c==$ ``blue") \{ $c(K')$:= "blue" \}  else  \{ $c(K')$:= "red"; \}\\
\hspace*{2em} 7.6.~~ return {\sc IS-B-Display}(Updated($M$), $v$);\\
~~~~~\}\vspace{0.5em}\\
\}\vspace{0.5em}\\
\hline     
\end{tabular}
%\caption{A SCCP algorithm derived from Lemmas 4\textendash8.}
%\label{SCCP_Alg}
%\end{figure}
}

\newpage

\renewcommand\thefigure{S\arabic{figure}}    

\setcounter{figure}{0}    
\section{Supplementary Figures}

   \begin{figure}[h!]
    \centering
    \includegraphics[scale=1]{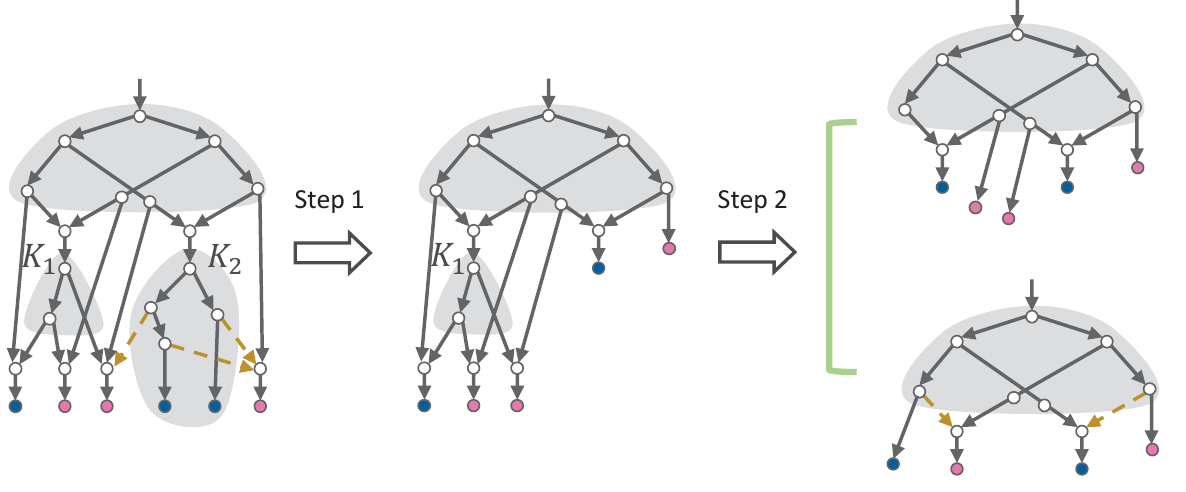}
    \caption{ Illustration of the {\sc CCP Algorithm One} in Section 3. The left panel contains the input network $N$. In recursive step 1,   $K_2$ is first selected in this example. Since $K_2$ is visible on two blue leaves,  we shall keep the blue leaves and switch off the red leaves below $K_2$ to obtain $N^{K_2}_{blue}$ (through removing the dashed edges), which is updated into the middle network. In recursive step 2, $K_1$ is selected. Since $K_1$ is invisible, we have to consider two possible cases: (1) blue leaves are kept in $K_1$, resulting in the right top network after updating; and (2) red leaves are kept in $K_1$, resulting in the right bottom network after updating. In recursive step 3, we consider both networks one by one. 
The cluster of blue leaves is not displayed in the right top network, but is displayed in the right bottom network. }
    \label{Fig4}
 \end{figure}

\newpage 

\begin{figure}[t!]
    \centering
    \includegraphics[scale=1]{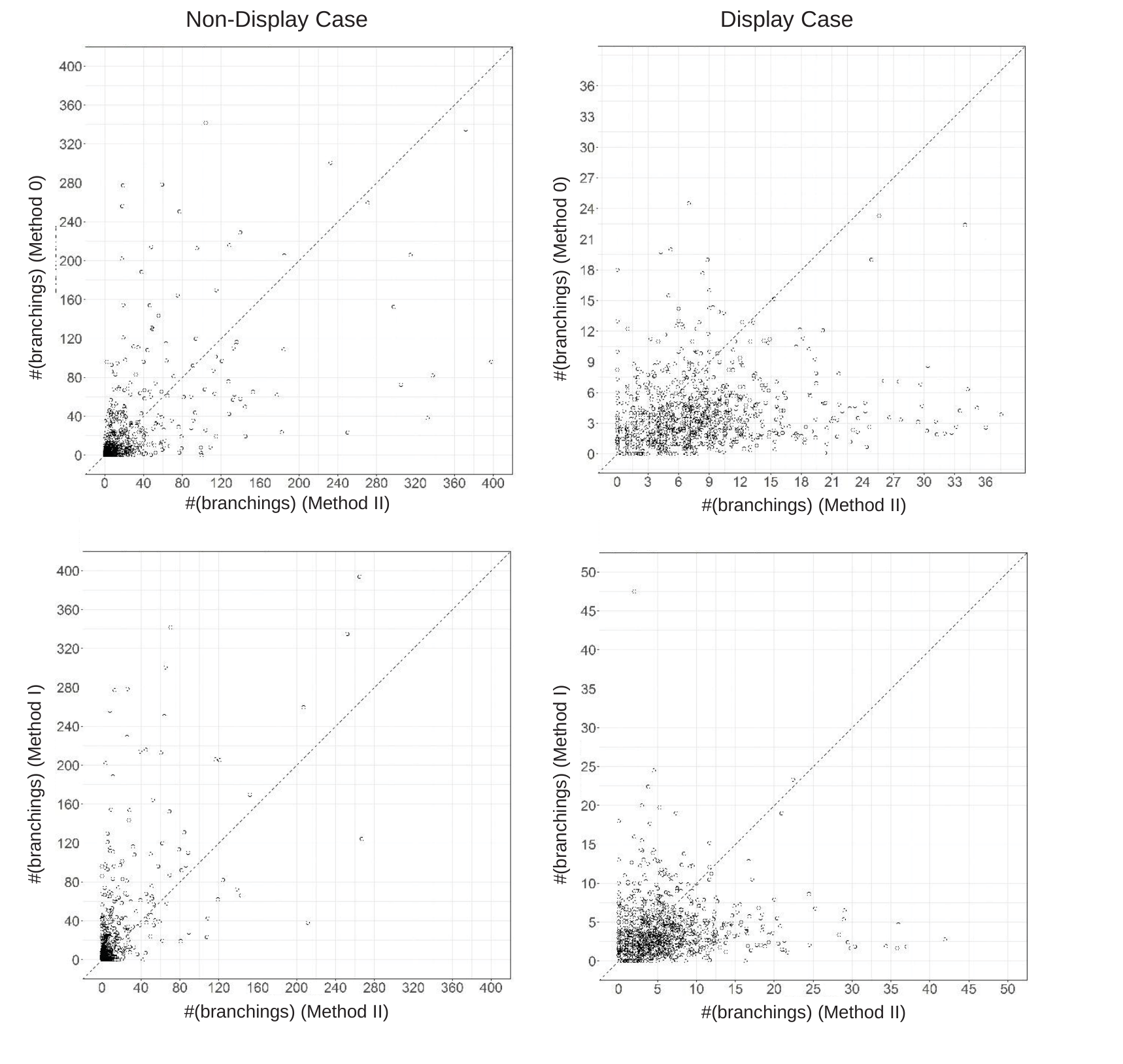}
    \caption{Performance comparison of the three methods on random networks in Group A. Method 0:  the method in \cite{Lu_17_APBC}; Method 1: the simple method presented in Section 3; 
Method 2, the method presented in section 5 (see remark in the section). \label{FigS2}}
\end{figure}

\newpage

\begin{figure}[t!]
    \centering
    \includegraphics[scale=1]{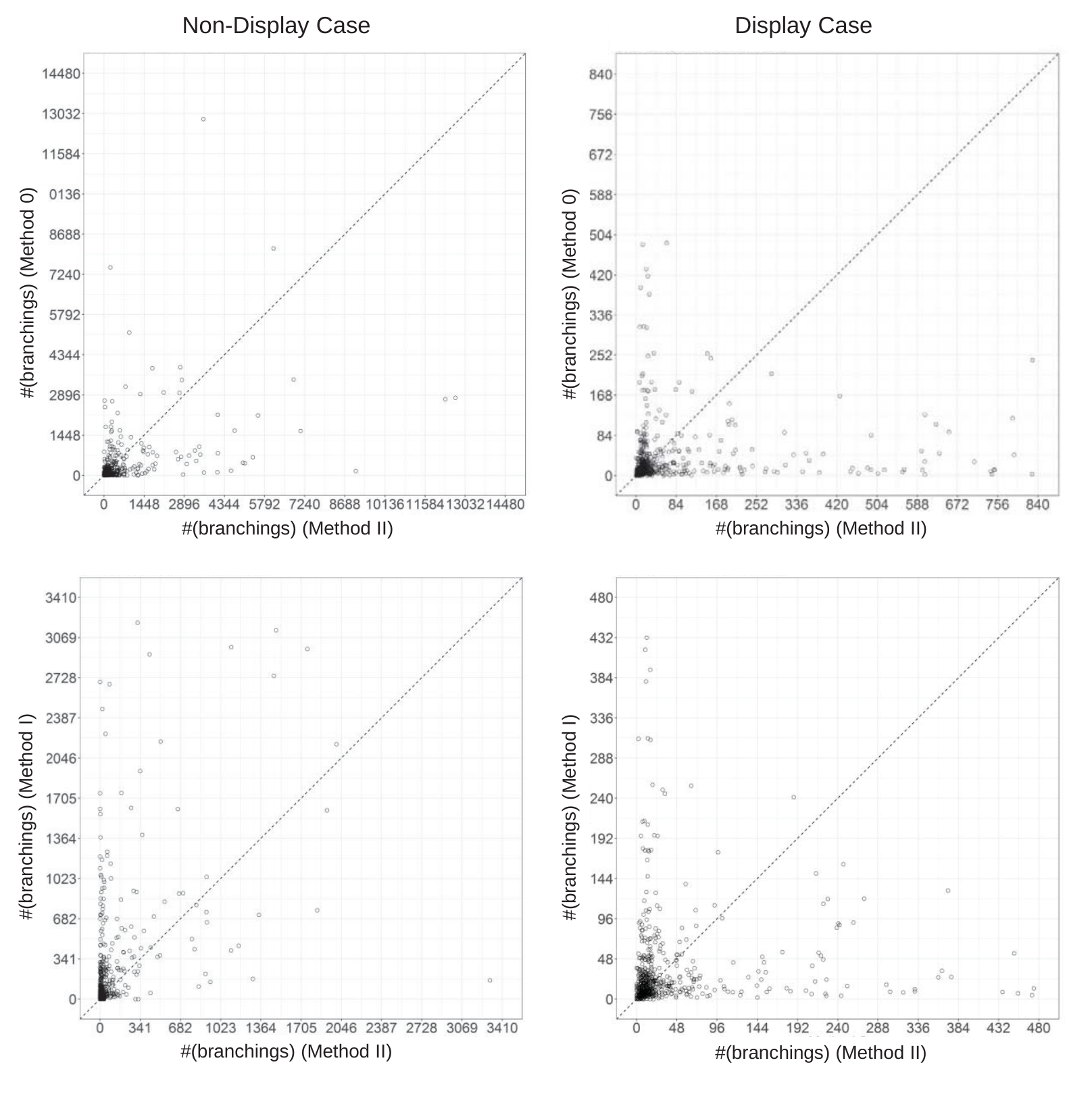}
    \caption{Performance comparison of the three methods on random networks in Group B. Method 0:  the method in \cite{Lu_17_APBC}; Method 1: the simple method presented in Section 3; 
Method 2, the method presented in section 5 (see remark in the section). \label{FigS3}}
\end{figure}

\end{document}